\documentclass{birkjour}
\usepackage{amssymb,bbm}

\def\N{\mathbbm N}

\def\Q{\mathbbm Q}
\def\R{\mathbbm R}
\def\C{\mathbbm C}
\def\Pf{\operatorname{Pf}}

\renewcommand{\leq}{\leqslant}
\renewcommand{\geq}{\geqslant}

\newtheorem{theorem}{Theorem}
\newtheorem{conjecture}[theorem]{Conjecture}

\begin{document}

\title[Advanced Computer Algebra for Determinants]{Advanced Computer Algebra for Determinants$\!$}

\author{Christoph Koutschan}
\address{%
Research Institute for Symbolic Computation (RISC)\\
Johannes Kepler University\\
Altenberger Stra\ss e 69\\
A-4040 Linz\\
Austria}
\email{koutschan@risc.jku.at}

\author{Thotsaporn~``Aek''~Thanatipanonda}
\address{%
Research Institute for Symbolic Computation (RISC)\\
Johannes Kepler University\\
Altenberger Stra\ss e 69\\
A-4040 Linz\\
Austria}
\email{thotsaporn@gmail.com}

\begin{abstract}
We prove three conjectures concerning the evaluation of determinants,
which are related to the counting of plane partitions and rhombus
tilings.  One of them was posed by George Andrews in 1980, the
other two were by Guoce Xin and Christian Krattenthaler. Our proofs
employ computer algebra methods, namely, the holonomic ansatz proposed
by Doron Zeilberger and variations thereof. These variations make
Zeilberger's original approach even more powerful and allow for
addressing a wider variety of determinants.  Finally, we
present, as a challenge problem, a conjecture about a closed-form
evaluation of Andrews's determinant.
\end{abstract}

\thanks{CK was supported by the Austrian Science Fund (FWF): P20162-N18,
and in part by the grant DMU 03/17 of the Bulgarian National Science Fund.
TT was supported by the strategic program ``Innovatives O\"{O} 2010plus'' by the Upper Austrian Government.\\
The final publication is available at www.link.springer.com
(Annals of Combinatorics, DOI~10.1007/s00026-013-0183-8).}



\maketitle

\section{Introduction}

The concept of determinants evolved as early as 1545 when Girolamo
Cardano tried to solve systems of linear equations.  The mathematics
community slowly realized the importance of determinants; we had to
wait for more than 200 years before someone formally defined the term
``determinant''.  It was first introduced by Carl Friedrich Gau\ss\ in
his \emph{Disquisitiones Arithmeticae} in 1801. Determinants possess
many nice properties and formulas such as multiplicativity, invariance
under row operations, Cramer's rule, etc. Every student nowadays
learns how to compute the determinant of a specific given matrix, say,
with fixed dimension and containing numeric quantities as entries. On
the other hand, there are lots of matrices with symbolic entries that
have a nice closed-form formula for arbitrary dimension. The first
example in the history of mathematics and still the most prominent one
is the Vandermonde matrix.

Starting in the mid 1970's, the importance of evaluating determinants
became even more significant when people related counting problems
from combinatorics to the evaluation of certain determinants.
Evaluating determinants of matrices whose dimension varies according
to a parameter is a delicate problem but at the same time a very
important one.  Around the same time the field of computer algebra
emerged.  Doron Zeilberger was amongst the first mathematicians to
realize the importance of computer algebra algorithms for
combinatorial problems, special function identities, symbolic
summation, and many more. In his paper~\cite{Zeilberger07} he built
the bridge between these two topics, namely, symbolic determinant
evaluation and computer algebra, and since then his \emph{holonomic
  ansatz} has been successfully applied to many problems related to
the evaluation of determinants. The most prominent one is probably a
notorious conjecture from enumerative combinatorics, the so-called
$q$-TSPP conjecture, which was the only remaining open problem from
the famous list~\cite{Stanley86} by Richard Stanley (it also appeared
in~\cite{Krattenthaler99}) until it was recently
proved~\cite{KoutschanKauersZeilberger11} using Zeilberger's holonomic
ansatz.

In this paper, we solve some of the problems that are listed in
Christian Krattenthaler's complement~\cite{Krattenthaler05} to his
celebrated essay~\cite{Krattenthaler99} (the attentive reader may
already have observed that our title is an allusion to this
reference).  At the same time we show that the holonomic ansatz can be
modified in various ways in order to apply it to particular problems
that could not be addressed with the original method.

\section{Zeilberger's Holonomic Ansatz}
\label{sec.Zeil}

For sake of self-containedness we recall briefly the original
holonomic ansatz for determinant evaluations as it was proposed by
Zeilberger~\cite{Zeilberger07}. Its steps are completely automatic and
produce a rigorous proof---provided that they can be successfully
carried out in the example at hand. In particular, the approach relies
on the existence of a ``nice'' description for an auxiliary function
(it appears as $c_{n,j}$ below); if such a description does not exist
then the holonomic ansatz fails.  That's why we call it an
``approach'' or an ``ansatz'', rather than an algorithm.

Generally speaking, Zeilberger's holonomic ansatz addresses
determinant evaluations of the type
\[
  \det A_n = \det_{1 \leq i,j \leq n}(a_{i,j})=b_n\qquad(n\geq 1)
\]
where the entries $a_{i,j}$ of the $n\times n$ matrix~$A_n$ and the
(conjectured) evaluation~$b_n$ (where $b_n\neq 0$ is required for all
$n\geq 1$) are explicitly given.  The underlying principle is an
induction argument on~$n$. The base case $a_{1,1}=b_1$ is easily
checked.  Now assume that the determinant evaluation has been proven
for $n-1$. In particular, it follows that $\det A_{n-1}$ is nonzero by
the general assumption on~$b_n$. Hence the rows of $A_{n-1}$ are
linearly independent and thus the linear system
\[
  \begin{pmatrix}
    a_{1,1} & \cdots & a_{1,n-1} & a_{1,n} \\
    \vdots & \ddots & \vdots & \vdots \\
    a_{n-1,1} & \cdots & a_{n-1,n-1} & a_{n-1,n} \\
    0 & \cdots & 0 & 1
  \end{pmatrix}
  \begin{pmatrix}
    c_{n,1} \\
    \vdots \\
    c_{n,n-1} \\
    c_{n,n}
  \end{pmatrix}
  =
  \begin{pmatrix}
    0 \\
    \vdots\\
    0\\
    1
  \end{pmatrix}
\]
has a unique solution. In the well-known Laplace expansion formula
(here with respect to the last row)
\[
  \det_{1 \leq i,j \leq n}(a_{i,j}) = \sum_{j=1}^n (-1)^{n+j}M_{n,j}a_{n,j}
\]
the expression $(-1)^{n+j}M_{n,j}$ is called the $(n,j)$-cofactor
of~$A_n$.  The minor $M_{n,j}$ is the determinant of the matrix
obtained by removing the $n$-th row and the $j$-th column. The above
linear system has been constructed in such a way that the
entry~$c_{n,j}$ in its solution is precisely the $(n,j)$-cofactor
of~$A_n$ divided by its $(n,n)$-cofactor (which is just the
$(n-1)$-determinant). This fact can easily be seen by Cramer's rule,
i.e., by considering the matrix $A_n^{(i)}$ that we obtain from $A_n$
by replacing the last row by the $i$-th row. For $1\leq i<n$, the
Laplace expansion of $A_n^{(i)}$ corresponds exactly to the $i$-th
equation in our linear system. Now the determinant of $A_n$ is given
by
\[
  b_{n-1}\sum_{j=1}^n c_{n,j} a_{n,j}.
\]
To complete the induction step it remains to show that this expression
is equal to~$b_n$.

The problem is that we cannot expect to obtain a closed-form
expression for the quantity~$c_{n,j}$ (otherwise, we certainly would be
able to derive a closed form for the determinant and we were done).
Instead, we will guess an implicit, recursive definition for a
bivariate sequence~$c_{n,j}$ and then prove that it satisfies
\begin{alignat}3
  c_{n,n}&=1&\qquad&(n\geq1),\label{eq.1}\\
  \sum_{j=1}^n c_{n,j}a_{i,j}&=0&&(1\leq i<n),\label{eq.2}\\
  \sum_{j=1}^n c_{n,j}a_{n,j}&=\frac{b_n}{b_{n-1}}&&(n\geq1)\label{eq.3}.
\end{alignat}
From the first two identities which correspond to the linear system
given above, it follows that our guessed~$c_{n,j}$ is indeed the
normalized $(n,j)$-cofactor. Identity~\eqref{eq.3} then certifies
that the determinant evaluates to~$b_n$. Hence the sequence~$c_{n,j}$
plays the r\^{o}le of a certificate for the determinant evaluation.

Now what kind of implicit definition for~$c_{n,j}$ could we think of?
Of course, there was a good reason for Zeilberger to name his approach
the ``holonomic ansatz''. That is to say, because he had the class of
holonomic functions (or better, sequences) in mind when he formulated
his approach. In short, this class consists of multi-dimensional
sequences that satisfy linear recurrence relations with polynomial
coefficients, such that the sequence is uniquely determined by
specifying \emph{finitely many} initial values (we omit some
additional technical conditions here). What makes the use of this
class of sequences convenient is the fact that it is not only closed
under the basic arithmetic operations (addition, multiplication), but
also under more advanced operations such as specialization,
diagonalization, and definite summation; all these are needed to prove
\eqref{eq.1}, \eqref{eq.2}, and \eqref{eq.3}. Moreover, there exist
known algorithms for performing zero tests and the previously
mentioned operations.  For more details on holonomic functions and
related algorithms see, e.g.,
\cite{Zeilberger90,Chyzak98,Koutschan09}). Thus if the matrix
entries~$a_{i,j}$ are holonomic and if luckily the auxiliary function
$c_{n,j}$ turns out to be holonomic, then the approach will succeed to
produce a recurrence for the quotient $b_n/b_{n-1}$.

\section{An Old Problem by George Andrews}
\label{sec.34}

In the context of enumerating certain classes of plane partitions,
namely, cyclically symmetric ones and descending ones, George
Andrews~\cite{Andrews80a} encountered an intriguing determinant which
he posed as a challenging problem. In Krattenthaler's
survey~\cite{Krattenthaler05}, it appears as \emph{Problem 34}.  This
conjecture does not even deal with a closed-form evaluation, but it is
``only'' about the quotient of two consecutive determinants; a
situation that strongly suggests to employ Zeilberger's holonomic
ansatz!

\begin{theorem}\label{thm.34}
Let the determinant~$D_1(n)$ be defined by
\[
  D_1(n) := \det_{1\leq i,j\leq n}\left(\delta_{i,j} + \binom{\mu+i+j-2}{j}\right)
\]
where~$\mu$ is an indeterminate and $\delta_{i,j}$
is the Kronecker delta function. Then the following relation holds:
\[
  \frac{D_1(2n)}{D_1(2n-1)} =
   (-1)^{(n-1)(n-2)/2}2^n\frac{
     \left(\frac12(\mu+2n)\right)_{\lfloor (n+1)/2\rfloor}\left(\frac12(\mu+4n+1)\right)_{n-1}}{
     (n)_n \left(\frac12(-\mu-4n+3)\right)_{\lfloor (n-1)/2\rfloor}}.
\]
\end{theorem}
\begin{proof}
By looking at the first few evaluations of~$D_1(n)$ (for $1\leq
n\leq8$ they are explicitly displayed in~\cite{Krattenthaler05}), we
see that only the quotient $D_1(2n)/D_1(2n-1)$ is nice, but not
$D_1(2n+1)/D_1(2n)$.  The reason is the occurrence of irreducible
nonlinear factors that change every two steps, i.e., $D_1(2n)$ and
$D_1(2n-1)$ share the same ``ugly'' factor (thus their quotient is
nice), but in $D_1(2n+1)$ the ``ugly'' part will be different (and
therefore the quotient $D_1(2n+1)/D_1(2n)$ does not factor nicely).

We first tried Zeilberger's original approach on the
determinant~$D_1(n)$. But we didn't even succeed to guess the
recurrences for~$c_{n,j}$ as they either are extraordinarily large or
do not exist at all. Moreover, we have good reasons to believe that
the quotient $D_1(2n+1)/D_1(2n)$ is not holonomic (see
Section~\ref{sec.chall}), in which case we know a priori that the
approach cannot succeed.

Therefore we have to come up with a variation of Zeilberger's
approach, that pays attention to even~$n$ only. This means that we
consider the normalized cofactors $c_{n,j}$ only for matrices of even
size. The new identities, to be verified, will be
\begin{alignat}{3}
  c_{n,2n} & = 1 &\qquad& (n\geq1), \tag{1a}\label{eq.1a}\\
  \sum_{j=1}^{2n} c_{n,j}a_{i,j} & = 0 && (1\leq i<2n), \tag{2a}\label{eq.2a}\\
  \sum_{j=1}^{2n} c_{n,j}a_{2n,j} & = \frac{b_{2n}}{b_{2n-1}} && (n\geq1). \tag{3a}\label{eq.3a}
\end{alignat}

In order to come up with an appropriate guess for the yet unknown
function~$c_{n,j}$, we compute the normalized cofactors for all
even-size matrices up to dimension~$30$. This gives a $15\times 30$
array with values in $\Q(\mu)$ that is used for guessing linear
recurrences for~$c_{n,j}$.  For this step, Kauers's Mathematica
package \texttt{Guess} has been employed, see~\cite{Kauers09} for more
details. To give the reader an impression of what the output looks
like, we display the results here in truncated form: Whenever the
abbreviation $\langle k\ \text{terms}\rangle$ appears, it indicates
that this polynomial cannot be factored into smaller pieces, and that
for better readability it is not displayed in full size here.

\begin{equation}\label{eq.rec1}
\begin{split}
  & 2n (j+1) (2n-1) (2j+\mu) (j-2n) (j-2n+1) \\
  & \qquad \times (\mu+4n-5) (\mu+4n-3) (j+\mu+2n-1) c_{n,j} =\\
  & j (j+\mu-1) (2j+\mu-1) (j-2n+3) (\mu+4n-3)\\
  & \qquad \times (j^4+2 j^3 \mu + \dots\langle 24\ \text{terms}\rangle + 12) c_{n-1,j+1} - {}\\
  & (j+1) (j+\mu+2n-3) (2j^6 \mu+8j^6n+ \dots\langle 92\ \text{terms}\rangle
  -210\mu n) c_{n-1,j}
\end{split}
\end{equation}

\begin{equation}\label{eq.rec2}
\begin{split}
  & (j-1) (j+\mu-3) (j+\mu-2) (2j+\mu-4) (j-2n) (j+\mu+2n-1) c_{n,j} =\!\\
  & j (j+\mu-3) (4j^4+8j^3\mu + \dots\langle 26\ \text{terms}\rangle + 16) c_{n,j-1} - {}\\
  & j (j-1) (j+\mu-2) (2j+\mu-2) (j-2n-2) (j+\mu+2n-3) c_{n,j-2}
\end{split}
\end{equation}

When translated into operator notation, the two recurrences~\eqref{eq.rec1}
and~\eqref{eq.rec2} constitute a Gr\"ob\-ner basis of the left ideal
that they generate in the corresponding operator algebra. As a
consequence it follows that they are compatible; this means that
starting from some given initial values, these recurrences will always
produce the same value for a particular $c_{n_0,j_0}$, independently
from the order in which they are applied. The support of the
recurrences suggests that fixing the initial values $c_{1,1}=-\mu/2$
and $c_{1,2}=1$ is sufficient: \eqref{eq.rec2} would produce $c_{1,j}$
for all $j>2$, and then \eqref{eq.rec1} could be used to obtain the full
array of values $(c_{n,j})_{n,j\geq 1}$.

Unfortunately it is not that easy, since there are two disturbing
pheno\-mena. The first is the factor $j-2n$ that appears in both
leading coefficients. Hence for computing $c_{n,2n}$ none of the
recurrences is applicable and we are stuck. Note that the factor
$j-2n+1$ in the leading coefficient of~\eqref{eq.rec1} is not a
problem since we can use~\eqref{eq.rec2} instead. We overcome the
problem with $c_{n,2n}$ by finding a recurrence of the form
\begin{equation}\label{eq.recdiag}
  p_2\,c_{n+2,j+4}+p_1\,c_{n+1,j+2}+p_0\,c_{n,j}=0,\qquad p_0,p_1,p_2\in\Q[n,j,\mu]
\end{equation}
in the ideal generated by \eqref{eq.rec1} and~\eqref{eq.rec2}. The
coefficient~$p_2$ does not vanish for $j=2n$, and thus, together with
the additional initial value~$c_{2,4}=1$, the recurrence \eqref{eq.recdiag} allows to
compute the values~$c_{n,2n}$.  

The second unpleasant phenomenon is the free parameter~$\mu$ which can
cause the same effect for certain choices.  For example, if $\mu=-6$
then we cannot compute $c_{2,3}$ from the given initial values since
both leading coefficients vanish simultaneously by virtue of the
factors $j-2n+1$ and $j+\mu+2n-1$. We handle it by restricting the
parameter~$\mu$ to the real numbers $>\!\mu_0$ for a certain
$\mu_0\in\R$, and by showing that all our calculations are sound under
this assumption (in most cases we use $\mu>0$ or $\mu>2$). But since
the determinant for every $n\in\N$ is a polynomial in~$\mu$ we can
extend our result afterwards to all $\mu\in\C$. Alternatively, one can
argue that $\mu$ is a formal parameter and therefore expressions like
$\mu+6$ are not zero as they are considered as elements in the
polynomial ring $\C[\mu]$.

We have now prepared the stage for executing the main part of
Zeilberger's approach. Identity~\eqref{eq.1a} is easily shown with
the help of recurrence~\eqref{eq.recdiag}: Substituting $j\to 2n$
produces a recurrence for the entries~$c_{n,2n}$, which together with
the initial values implies that $c_{n,2n}=1$ for all~$n$.
Identities~\eqref{eq.2a} and \eqref{eq.3a} are proven automatically as
well, since they are standard applications of holonomic closure
properties and summation techniques (creative telescoping). For these
tasks we have used the first author's Mathematica package
\texttt{HolonomicFunctions}~\cite{Koutschan10b}. The interested reader
is invited to study our computations in detail, by downloading the
electronic supplementary material from our
webpage~\cite{KoutschanThanatipanonda12a}.
\end{proof}

\section{Interlude: The Double-Step Variant}
\label{sec.double}

In this section, we propose a variant of Zeilberger's holonomic ansatz,
described in Section~\ref{sec.Zeil}, that enlarges the class of
determinants which can be treated by this kind of ansatz. The
condition $b_n\neq0$ for all $n\geq1$ imposes already some
restriction. For example, when studying a Pfaffian~$\Pf(A)$ for some
skew-symmetric matrix~$A$, one could be tempted to apply Zeilberger's
approach to the determinant of~$A$; recall that
$\Pf(A)^2=\det(A)$. The problem then is that $\det(A)$ is zero
whenever the dimension of~$A$ is odd.  Hence one would like to study
the quotient $b_{n}/b_{n-2}$ instead of the forbidden $b_n/b_{n-1}$;
as in the previous section, $n$ has to be restricted to the even
integers.  This dilemma can be solved by modifying Zeilberger's ansatz
subject to the Laplace expansion of Pfaffians,
see~\cite{IshikawaKoutschan12}.

On the other hand there are determinants which do factor nicely for
even dimensions but not for odd ones. Also here, we expect the
quotient $b_n/b_{n-2}$ to be nice, whereas the expression
$b_n/b_{n-1}$ might not even satisfy a linear recurrence and hence
could not be handled by Zeilberger's holonomic ansatz at
all. Similarly when the closed form~$b_n$ is different for even and
odd~$n$: While here the original approach could probably work in
principle, one may not succeed because of the computational complexity
that is caused by studying the quotient $b_n/b_{n-1}$, which is
expected to be more complicated than $b_n/b_{n-2}$. See
Theorems~\ref{thm.35} and~\ref{thm.36} below for such examples, which
actually have motivated us to propose the following variant.

As we announced we now generalize Zeilberger's method in order to
produce a recurrence for the quotient of determinants whose dimensions
differ by two. As before, let $M=(a_{i,j})_{1\leq i,j\leq n}$ and let
$b_n$ denote the conjectured evaluation of~$\det(M)$, which for all
$n$ in question has to be nonzero.  We pull out of the hat two
discrete functions $c'_{n,j}$ and $c''_{n,j}$ and verify the following
identities:

\begin{equation}\tag{1b}\label{eq.1b}
  c'_{n,n-1} = c''_{n,n} = 1,\qquad c'_{n,n} = c''_{n,n-1} = 0,
\end{equation}
\begin{equation}\tag{2b}\label{eq.2b}
  \sum_{j=1}^n a_{i,j}c'_{n,j} = \sum_{j=1}^n a_{i,j}c''_{n,j} = 0\qquad (1\leq i\leq n-2),
\end{equation}
\begin{equation}\tag{3b}\label{eq.3b}
  \bigg(\sum_{j=1}^n a_{n-1,j}\,c'_{n,j}\bigg)\!\bigg(\sum_{j=1}^n a_{n,j}\,c''_{n,j}\bigg)
  -\bigg(\sum_{j=1}^n a_{n-1,j}\,c''_{n,j}\bigg)\!\bigg(\sum_{j=1}^n a_{n,j}\,c'_{n,j}\bigg)
  = \frac{b_n}{b_{n-2}}.
\end{equation}
Then the determinant evaluation follows as a consequence, using a
similar induction argument as in Section~\ref{sec.Zeil}.

Let's try to give the motivation for this approach which also explains
why it works. The idea is based on the formula for the determinant of
a block matrix:
\[
  \det(M) = \det\left(\!\!\begin{array}{cc}M_1 & M_2\\ M_3 & M_4\end{array}\!\!\right)=
  \det(M_1)\det(M_4-M_3M_1^{-1}M_2).
\]
We want to divide the matrix~$M$ into blocks such that $M_4$ is a
$2\times 2$ matrix.  Let $C=(C',C'')$ denote the $(n-2)\times 2$
matrix whose first column is $C'=(c'_{n,j})_{1\leq j\leq n-2}$ and
whose second column is $C''=(c''_{n,j})_{1\leq j\leq n-2}$. With this
notation, the two equations~\eqref{eq.2b} can be written as
\[
  (M_1,M_2)\left(\!\!\begin{array}{cc}C' & C''\\ 1 & 0\\ 0 & 1\end{array}\!\!\right) =
  (M_1,M_2)\left(\!\!\begin{array}{c}C\\ I_2\end{array}\!\!\right) = M_1C+M_2 = 0,
\]
where the conditions of Equation~\eqref{eq.1b} have been employed to
constitute the identity matrix~$I_2$. Now by the induction hypothesis
we may assume that $\det_{1\leq i,j\leq n-2}(a_{i,j}) = \det(M_1)$ equals
$b_{n-2}$, which by our general assumption is nonzero.  Thus the above
system determines~$C$ uniquely and we can write $C=-M_1^{-1}M_2$.

Finally, we obtain the missing part from the block matrix formula that
gives us the quotient $\det(M)/\det(M_1)$, as the determinant of the
$2\times 2$ matrix
\[
  (M_3,M_4)\left(\!\!\begin{array}{c}C\\ I_2\end{array}\!\!\right) =
  M_3C+M_4 = M_4-M_3M_1^{-1}M_2.
\]
This is exactly what is expressed in Equation~\eqref{eq.3b}, with
all matrix multiplications explicitly written out.

Even though it turned out during our research that there is no need to
apply the double-step method in the present context, we decided to
keep it, as we are sure that it will be useful for future use.  Let us
also remark that it is straight-forward to generalize this idea in
order to produce a triple-step method, etc. But since the identities
to be verified become more and more complicated, we don't believe that
these large-step methods are relevant in practice.

\section{Small Change, Big Impact}
\label{sec.35}

The next determinants we want to study appear as \emph{Conjecture~35}
and \emph{Conjecture~36} in~\cite{Krattenthaler05}.
Krattenthaler (private communication) describes the story of how they
were raised:
\begin{quote}
I wrote this article during a stay at the Mittag-Leffler Institut.
Guoce Xin was also there, as well as Alain Lascoux.  They followed the
progress on the article with interest.  So, it was Guoce Xin, who had
been looking at similar determinants at the time (in the course of his
work with Ira Gessel on his big article on determinants and path
counting~\cite{GesselXin06}), who told me what became Conjecture~35. I
made some experiments and discovered Conjecture~36. Alain Lascoux saw
all this, and he came up with Conjecture~37.
\end{quote}

Xin's observation was that a certain matrix, very similar to
the one of Section~\ref{sec.34}, has a determinant that factors
completely; the only change is the sign of the Kronecker
delta~$\delta_{i,j}$. But still, the evaluation is given as a case
distinction between even and odd dimensions of the matrix.

\begin{theorem}\label{thm.35}
Let~$\mu$ be an indeterminate and $n$ be a nonnegative integer. Then
the determinant
\begin{equation}\label{eq.det35}
  \det_{1\leq i,j\leq n}\left(-\delta_{i,j}+\binom{\mu+i+j-2}{j}\right)
\end{equation}
equals
\begin{align*}
  & (-1)^{n/2} \, 2^{n(n+2)/4} \,
     \frac{\left(\frac{\mu}{2}\right)_{n/2}}{\left(\frac{n}{2}\right)!} \,
     \Bigg(\prod_{i=0}^{(n-2)/2} \frac{i!^2}{(2i)!^2}\Bigg) \\
  & \quad\times\Bigg(\prod_{i=1}^{\lfloor n/4\rfloor}
      \Big({\textstyle \frac12(\mu+6i-1)}\Big)_{(n-4i+2)/2}^2 \,
      \Big({\textstyle \frac12(-\mu-3n+6i)}\Big)_{(n-4i)/2}^2\Bigg)
\end{align*}
if $n$ is even, and it equals
\begin{align*}
  & (-1)^{(n-1)/2} \, 2^{(n+3)(n+1)/4} \,
     \Big({\textstyle\frac12(\mu-1)}\Big)_{(n+1)/2} \, \Bigg(\prod_{i=0}^{(n-1)/2} \frac{i!(i+1)!}{(2i)!(2i+2)!}\Bigg)\\
  & \quad\times \Bigg(\prod_{i=1}^{\lfloor(n+1)/4\rfloor} \!\!
      \Big({\textstyle \frac12(\mu+6i-1)}\Big)_{(n-4i+1)/2}^2 \,
      \Big({\textstyle \frac12(-\mu-3n+6i-3)}\Big)_{(n-4i+3)/2}^2\Bigg)
\end{align*}
if $n$ is odd.
\end{theorem}
\begin{proof}
We could try to solve this problem directly, either by Zeilberger's
original ansatz or in the way we did in Section~\ref{sec.34}. However,
this does not work in practice as the computations become too large
(in the second case we were even able to guess the recurrences for
$c_{n,j}$, but their size destroyed any hope to prove \eqref{eq.2a}
and~\eqref{eq.3a}). Next, we could give the double-step method a try,
which we described in Section~\ref{sec.double}. We succeeded to make
the proof for some concrete integer~$\mu$, but again, it seems
intractable for symbolic~$\mu$. Instead we make a small detour and
break this determinant into pieces in order to make the calculations
smaller. To obtain the desired result, we put these pieces together by
the Desnanot-Jacobi adjoint matrix theorem, also known as Dodgson's
rule; this formula gave rise to a celebrated algorithmic determinant
evaluation method as well~\cite{Zeilberger96,AmdeberhanZeilberger01},
but that approach is different from our usage of Dodgson's rule.  Let
us introduce the following notation
\[
  b_n(I,J) := b_n(I,J,\mu) := \det_{\genfrac{}{}{0pt}{}{I\leq i \leq n-1+I}{J \leq j \leq n-1+J}}
    \left(-\delta_{i,j}+\binom{\mu+i+j-2}{j}\right)
\]
so that our determinant~\eqref{eq.det35} is denoted by~$b_n(1,1)$.
In this notation the Desnanot-Jacobi identity is stated as follows:
\begin{equation}\label{eq.dj}
  b_{n}(0,0)b_{n-2}(1,1) = b_{n-1}(0,0)b_{n-1}(1,1) - b_{n-1}(0,1)b_{n-1}(1,0).
\end{equation}
After substituting $n\to 2n+2$ and $n\to 2n+1$ in~\eqref{eq.dj} and using
the fact that $b_{2n}(0,0)=-b_{2n-1}(1,1)$ (to be shown later) we have
\begin{align*}
  b_{2n+1}(1,1) & = \frac{b_{2n+1}(0,1)b_{2n+1}(1,0)}{b_{2n}(1,1)+b_{2n+1}(0,0)}\\
  b_{2n-1}(1,1) & = \frac{-b_{2n}(0,1)b_{2n}(1,0)}{b_{2n}(1,1)+b_{2n+1}(0,0)}
\end{align*}
from which the desired quotient can be obtained:
\begin{equation}\label{eq.quo1}
  \frac{b_{2n+1}(1,1)}{b_{2n-1}(1,1)} = 
    -\frac{b_{2n+1}(0,1)}{b_{2n}(0,1)}\cdot\frac{b_{2n+1}(1,0)}{b_{2n}(1,0)}.
\end{equation}
Similarly, we substitute $n\to 2n+1$ and $n\to 2n$ into~\eqref{eq.dj}
and use the fact that $b_{2n-1}(0,0)=0$ (again, to be shown later), to derive
the quotient of even determinants
\begin{equation}\label{eq.quo2}
  \frac{b_{2n}(1,1)}{b_{2n-2}(1,1)} = 
    -\frac{b_{2n}(0,1)}{b_{2n-1}(0,1)}\cdot\frac{b_{2n}(1,0)}{b_{2n-1}(1,0)}.
\end{equation}
We now use the first variation of Zeilberger's ansatz (see
Section~\ref{sec.34}) to derive recurrences for the quotients
$b_{2n+1}(0,1)/b_{2n}(0,1)$, etc.  which appear on the right-hand
sides of \eqref{eq.quo1} and~\eqref{eq.quo2}.  Since the
arguments closely follow the lines of the proof of
Theorem~\ref{thm.34}, we don't detail further this part and refer
to the electronic material~\cite{KoutschanThanatipanonda12a}.  
Although for our purposes it would suffice to work with these
recurrences, we succeed in solving them in closed form:
\begin{align}
Q_1(n) := \frac{b_{2n+1}(0,1)}{b_{2n}(0,1)} & = \label{eq.Q1}
  \frac{2\, \big(\frac{\mu}{2} +2n\big)_{\!n+1} \big(\mu+2n-1\big)_{\!n+1}}
       {\big(n+2\big)_{\!n+1} \big(\frac{\mu}{2}+n\big)_{\!n+1}},\\
Q_2(n) := \frac{b_{2n}(0,1)}{b_{2n-1}(0,1)} & = \label{eq.Q2}
  \frac{(\mu+2n-2) \big(\frac{\mu}{2}+2n-1\big)_{\!n-1} \big(\mu+2n+1\big)_{\!n-1}}
       {n\,\big(n\big)_{\!n-1} \big(\frac{\mu}{2}+n+1\big)_{\!n-1}},\\
Q_3(n) := \frac{b_{2n+1}(1,0)}{b_{2n}(1,0)} & = \label{eq.Q3}
  \frac{2\,\big(\frac{\mu}{2}+2n\big)_{\!n+1} \big(\mu+2n+1\big)_{\!n-1}}
       {\big(n+1\big)_{\!n} \big(\frac{\mu}{2}+n+1\big)_{\!n}},\\
Q_4(n) := \frac{b_{2n}(1,0)}{b_{2n-1}(1,0)} & = \label{eq.Q4}
  \frac{2\,\big(\frac{\mu}{2}+2n-1\big)_{\!n-1} \big(\mu+2n+1\big)_{\!n-1}}
       {\big(n\big)_{\!n-1} \big(\frac{\mu}{2}+n+1\big)_{\!n-1}}.
\end{align}
These quotients immediately give closed-form evaluations of the
corresponding determinants (see also Theorems~\ref{thm.b01}
and~\ref{thm.b10}).  It remains to justify the assumptions
$b_{2n}(0,0)=-b_{2n-1}(1,1)$ and $b_{2n-1}(0,0)=0$ that were used to
derive \eqref{eq.quo1} and~\eqref{eq.quo2}.

In order to evaluate the quotient $b_{2n}(0,0)/b_{2n-1}(1,1)$ we need
to modify the method presented in Section~\ref{sec.34}: We apply
Laplace expansion with respect to the first row of the matrix, instead
of the $n$-th row, and we normalize the auxiliary function~$c_{n,j}$
in such a way that $c_{n,0}=1$. If we come up with a recursive
description of some function~$c_{n,j}$ and are able to verify the
following identities, then we are done:
\begin{alignat}{3}
  c_{n,0} & = 1 && (n\geq 1),\tag{1c}\label{eq.1c}\\
  \sum_{j=0}^{2n-1}c_{n,j}a_{i,j} & = 0 && (0<i\leq 2n-1), \tag{2c}\label{eq.2c}\\
  \sum_{j=0}^{2n-1}c_{n,j}a_{0,j} & = \frac{b_{2n}(0,0)}{b_{2n-1}(1,1)} = -1 &\qquad& (n\geq 1). \tag{3c}\label{eq.3c}
\end{alignat}
As before, the details can be found in~\cite{KoutschanThanatipanonda12a}.

Last but not least we have to argue that $b_{2n-1}(0,0)$ vanishes.
Krattenthaler kindly pointed us to
\cite[Theorem~11]{CiucuEisenkoeblKrattenthalerZare01} (see also
\cite[Theorem~35]{Krattenthaler99}) which contains this statement. Anyway we
have produced an alternative, computerized proof: Actually it is very
simple, since we just have to come up with a guess for the
coefficients of a linear combination of the columns (or rows) that
gives 0, and then prove that our guess does the job. Hence we find a
recursive description of a function $c_{n,j}$ for $n\geq 1$ and $0\leq j\leq
2n-2$, such that
\[
  c_{n,0}\begin{pmatrix} a_{0,0} \\ \vdots \\ a_{2n-2,0} \end{pmatrix} + \dots + 
  c_{n,2n-2}\begin{pmatrix} a_{0,2n-2} \\ \vdots \\ a_{2n-2,2n-2} \end{pmatrix}
  = \begin{pmatrix} 0 \\ \vdots \\ 0 \end{pmatrix}
\]
and such that there is an index~$j$ for which $c_{n,j}\neq 0$.  For
guessing, we compute $c_{n,j}$ for all~$n$ up to some bound and
normalize them. Luckily the nullspace of the above system has always
dimension~$1$, otherwise it would be trickier to find a suitable
linear combination (however, in our problem this is no surprise since
we are finally aiming at proving that the minor $M_{0,0}$ of this
matrix evaluates to some nonzero expression).  So we are successful in
guessing the recurrences of $c_{n,j}$ and use them to prove
\begin{alignat*}{3}
  c_{n,2n-2} & = 1 &\qquad& (n\geq 1),\\
  \sum_{j=0}^{2n-2}c_{n,j}a_{i,j} & = 0 && (0\leq i\leq 2n-2).
\end{alignat*}
This concludes the proof of Theorem~\ref{thm.35}.
\end{proof}
Since the evaluations of the determinants $b_n(0,1)$ and $b_n(1,0)$
are interesting results in their own right, but are somehow hidden in
the proof of Theorem~\ref{thm.35}, we are going to state them
explicitly here.  
\begin{theorem}\label{thm.b01}
Let~$\mu$ be an indeterminate and $n$ be a nonnegative integer.  Let
$Q_1(n)$ and $Q_2(n)$ be defined as in~\eqref{eq.Q1}
and~\eqref{eq.Q2}, respectively. Then the determinant
\[
  b_n(0,1) = 
  \det_{1\leq i,j\leq n}\left(-\delta_{i-1,j}+\binom{\mu+i+j-3}{j}\right)
\]
equals 
\[
  \Bigg(\prod_{k=0}^{n/2-1} Q_1(k)\Bigg) \Bigg(\prod_{k=1}^{n/2} Q_2(k)\Bigg)
\]
if $n$ is even, and it equals
\[
  (\mu-1) \prod_{k=1}^{(n-1)/2} \Big(Q_1(k) Q_2(k)\Big)
\]
if $n$ is odd.
\end{theorem}
\begin{proof}
It is analogous to the proof of Theorem~\ref{thm.34} and can be found
in~\cite{KoutschanThanatipanonda12a}.
\end{proof}
\begin{theorem}\label{thm.b10}
Let~$\mu$ be an indeterminate and $n$ be a nonnegative integer.  Let
$Q_3(n)$ and $Q_4(n)$ be defined as in~\eqref{eq.Q3}
and~\eqref{eq.Q4}, respectively. Then the determinant
\[
  b_n(1,0) = 
  \det_{1\leq i,j\leq n}\left(-\delta_{i,j-1}+\binom{\mu+i+j-3}{j-1}\right)
\]
equals
\[
  \Bigg(\prod_{k=0}^{n/2-1} Q_3(k)\Bigg) \Bigg(\prod_{k=1}^{n/2} Q_4(k)\Bigg)
\]
if $n$ is even, and it equals
\[
  \prod_{k=1}^{(n-1)/2} \Big(Q_3(k) Q_4(k)\Big)
\]
if $n$ is odd.
\end{theorem}
\begin{proof}
It is analogous to the proof of Theorem~\ref{thm.34} and can be found
in~\cite{KoutschanThanatipanonda12a}.
\end{proof}

By replacing $j$ by $j+1$ at the bottom of the binomial coefficient in
the entries of~\eqref{eq.det35}, we arrive at our last determinant; it
has been discovered by Krattenthaler and appears as
\emph{Conjecture 36} in his paper~\cite{Krattenthaler05} (note that
the formula there is erroneous, as one of the product quantors is
missing so that the corresponding factor slided into the previous
product). Also this problem has its own combinatorial interpretation
in terms of counting certain rhombus tilings.

\begin{theorem}\label{thm.36}
Let~$\mu$ be an indeterminate. For any odd nonnegative integer~$n$ there holds
\begin{align*}
  & \det_{1\leq i,j\leq n}\left(-\delta_{i,j}+\binom{\mu+i+j-2}{j+1}\right) = {}\\
  & \qquad (-1)^{(n-1)/2} \, 2^{(n-1)(n+5)/4} \, (\mu+1) \,
    \frac{\left(\frac12(\mu-2)\right)_{(n+1)/2}}{\left(\frac12(n+1)\right)!} \,
    \Bigg(\prod_{i=0}^{(n-1)/2} \frac{i!^2}{(2i)!^2}\Bigg)\\
  & \qquad\times\Bigg(\prod_{i=1}^{\lfloor(n+3)/4\rfloor} \Big({\textstyle\frac12(\mu+6i-3)}\Big)_{(n-4i+3)/2}^2\Bigg)\\
  & \qquad\times\Bigg(\prod_{i=1}^{\lfloor(n+1)/4\rfloor} \Big({\textstyle\frac12(-\mu-3n+6i-1)}\Big)_{(n-4i+1)/2}^2\Bigg).
\end{align*}
\end{theorem}
\begin{proof}
Because of the similarity of this determinant with~\eqref{eq.det35}, we are
able to relate these two problems via shifting the starting points:
\[
  \begin{split}
  & \det_{1\leq i,j\leq 2n-1}\left(-\delta_{i,j}+\binom{\mu+i+j-2}{j+1}\right) = \\
  & \det_{2 \leq i,j\leq 2n}\left(-\delta_{i,j}+\binom{(\mu-2)+i+j-2}{j}\right).
  \end{split}
\]
By using the notation from above, the determinant of
Theorem~\ref{thm.36} is denoted by $b_{2n-1}(2,2,\mu-2)$. Analogously
to \eqref{eq.1c}--\eqref{eq.3c}, we apply a variation of Zeilberger's
approach to derive a recurrence for
\[
  q_n(\mu) = \frac{b_{2n}(1,1,\mu)}{b_{2n-1}(2,2,\mu)}.
\]
The result is $q_{n+1}(\mu)-q_n(\mu)=0$ which reveals that the
quotient~$q_n(\mu)$ is constant. Together with the initial value
$q_1(\mu)=-4/(\mu+3)$ and the fact that $b_{2n}(1,1,\mu)$ is already
known from Theorem~\ref{thm.35}, we get the desired result. Once
again, we refer to~\cite{KoutschanThanatipanonda12a} for the details
of the computations.
\end{proof}

\noindent
As mentioned above, Lascoux found that the more general determinant
\[
  \det_{1\leq i,j\leq n}\left(-\delta_{i,j+r-1} + \binom{\mu+i+j-2}{j+r-1}\right)
\]
factors completely for odd natural numbers $n$ and~$r$, and its
complicated evaluation, which was figured out by
Krattenthaler, appears as \emph{Conjecture~37}
in~\cite{Krattenthaler05}. We remark that the formula given there
holds only for $r\leq n$; otherwise the Kronecker delta does not show
up in the matrix and the evaluation is much simpler.  We cannot attack
this determinant directly with Zeilberger's ansatz since the matrix
entries do not evaluate to polynomials in~$\mu$ for concrete integers
$i$ and~$j$, as long as $r$ is kept symbolically.  Therefore the
guessing for $c_{n,j}$ will not work. A different strategy would
consist in finding some connection between the cases $r$ and $r+2$;
then induction on $r$ would provide a proof, using
Theorem~\ref{thm.35} as the base case $r=1$. Unfortunately we were not
able to achieve this goal.

\section{A Challenge Problem}
\label{sec.chall}
We want to conclude our article with a challenge problem for the next
generation of computer algebra tools. In Section~\ref{sec.34}, we have
only proven a statement about the quotient of two consecutive
determinants (Theorem~\ref{thm.34}). But so far nobody has come up
with a closed form for the determinant~$D_1(n)$. We now present a
conjectured closed form, which, however, we are unable to prove with
the methods described in the present paper. We have already remarked
in Section~\ref{sec.34} that the quotient $D_1(n)/D_1(n-1)$ most
probably is not holonomic; in that case Zeilberger's holonomic ansatz
and our variations of it are not applicable.

\begin{conjecture}\label{conj.34}
Let $\mu$ be an indeterminate and let the sequences $C(n)$, $F(n)$,
and $G(n)$ be defined as follows:
\begin{align*}
C(n) & = \frac{(-1)^n+3}{2}\prod_{i=1}^n\frac{\left\lfloor\frac{i}{2}\right\rfloor !}{i!},\\
E(n) & = (\mu+1)_n 
  \Bigg(\prod_{i=1}^{\left\lfloor\frac32\left\lfloor\frac12(n-1)\right\rfloor-2\right\rfloor} \!\!
    \Big(\mu+2i+6\Big)^{2\left\lfloor\frac13(i+2)\right\rfloor}\Bigg)\\
  & \qquad \times
  \Bigg(\prod_{i=1}^{\left\lfloor\frac32\left\lfloor\frac{n}{2}\right\rfloor-2\right\rfloor} \!\!
    \Big(\mu+2i+{\textstyle 2\left\lfloor\frac32\left\lfloor\frac{n}{2}+1\right\rfloor\right\rfloor}-1
    \Big)^{2\left\lfloor\frac12\left\lfloor\frac{n}{2}\right\rfloor-\frac13(i-1)\right\rfloor-1}\Bigg),\\
F_m(n) & = \Bigg(\prod_{i=1}^{\left\lfloor\frac14(n-1)\right\rfloor} \!\!\! (\mu+2i+n+m)^{1-2i-m}\Bigg) \\
  & \qquad \times
  \Bigg(\prod_{i=1}^{\left\lfloor\frac{n}{4}-1\right\rfloor} \! (\mu-2i+2n-2m+1)^{1-2i-m}\Bigg),
\end{align*}\begin{align*}
F(n) & = \begin{cases}  
    E(n) F_0(n), & \text{if}\ n\ \text{is even},\\
    \displaystyle E(n) F_1(n) \prod_{i=1}^{\frac12(n-5)} \!\! (\mu+2i+2n-1), & \text{if}\ n\ \text{is odd},
  \end{cases}\\
T(k) & = 55296k^6 + 41472(\mu-1)k^5 + 384(30\mu^2-66\mu+53)k^4\\
  & \qquad +96(\mu-1)(15\mu^2-42\mu+61)k^3\\
  & \qquad +4(19\mu^4-122\mu^3+419\mu^2-544\mu+72)k^2\\
  & \qquad +(\mu-1)(\mu^4-14\mu^3+101\mu^2-160\mu-84)k\\
  & \qquad +2(\mu-3)(\mu-2)(\mu-1)(\mu+1),\\
S_1(n) & = \sum_{k=1}^{n-1} {\textstyle
    \Big( 2^{6k} (\mu+8k-1) \left(\frac12\right)_{2k-1}^2 \left(\frac12(\mu+5)\right)_{2k-3} \left(\frac12(\mu+4k+2)\right)_{k-2} }\\
  & \qquad\quad \times {\textstyle 
    \left(\frac12(\mu+4k+2)\right)_{2n-2k-2} T(k) \Big)\Big/\Big( (2k)! \left(\frac12 (\mu+6k-3)\right)_{3k+4} \Big) },\\
S_2(n) & = \sum_{k=1}^{n-1} {\textstyle 
    \Big( 2^{6k} (\mu+8k+3) \left(\frac12\right)_{2k}^2 \left(\frac12(\mu+5)\right)_{2k-2} \left(\frac12(\mu+4k+4)\right)_{k-2} }\\
  & \qquad\quad \times {\textstyle
    \left(\frac12(\mu+4k+4)\right)_{2n-2k-2} T\!\left(k+\frac12\right) \!\Big) \Big/
    \Big( (2k+1)! \left(\frac12(\mu+6k+1)\right)_{3k+5} \Big) },\\
P_1(n) & = 2^{3n-1}\frac{\left(\frac12 (\mu+6n-3)\right)_{3n-2}}{\left(\frac12(\mu+5)\right)_{2n-3}}
  \left(\frac{\left(\frac12(\mu+2)\right)_{2n-2}}{(\mu+3)^2} + \frac{\mu(\mu-1)S_1(n)}{2^{13}}\right),\\
P_2(n) & = 2^{3n-1}\frac{\left(\frac12 (\mu+6n+1)\right)_{3n-1}}{\left(\frac12(\mu+5)\right)_{2n-2}}
  \left(\frac{(\mu+14) \left(\frac12(\mu+4)\right)_{2n-2}}{(\mu+7)(\mu+9)} + \frac{\mu(\mu-1)S_2(n)}{2^{9}}\right),\\
G(n) & = \begin{cases}
  P_1\!\left(\frac12(n+1)\right), & \text{if}\ n\ \text{is odd},\\
  P_2\!\left(\frac{n}{2}\right), & \text{if}\ n\ \text{is even}.
  \end{cases}
\end{align*}
Then for every positive integer~$n$ we have
\[
  \det_{1\leq i,j\leq n}\left(\delta_{i,j} + \binom{\mu+i+j-2}{j}\right) = 
  \textstyle C(n) F(n) G\!\left(\left\lfloor\frac12(n+1)\right\rfloor\right).
\]
\end{conjecture}

Let us add a few remarks on our conjectured closed form.  The elements
of the sequence~$C(n)$ are rational numbers of the form $1/k$ where
$k$ is an integer. The sequences $F(n)$ and $G(n)$ consist of monic
polynomials in~$\mu$ with integer coefficients.  The $F(n)$ factor
completely into linear factors of the form $(\mu+k)$ where $k\in\N$,
and thus have positive coefficients.  The $G(n)$ have positive
coefficients as well, but turn out to be mostly irreducible; the only
counterexample we found is
$G(4)=(\mu+34)(\mu^3+47\mu^2+954\mu+5928)$.  They correspond to the
``ugly factors'' mentioned in Section~\ref{sec.34}. For the
convenience of the reader, we provide the Mathematica code for all
quantities introduced in Conjecture~\ref{conj.34} in the supplementary
material~\cite{KoutschanThanatipanonda12a}.

In order to come up with this complicated conjecture, we computed the
determinants $D_1(n)$ for $1\leq n\leq 295$ which gave us the first
$148$ polynomials of the sequence~$G(n)$. These data enabled us to
guess recurrences for the subsequences~$P_1(n)$ and~$P_2(n)$; the
recurrences by the way are used in~\cite{KoutschanThanatipanonda12a}
to provide a fast procedure for computing~$D_1(n)$. The Maple package
\texttt{LREtools} was able to find ``closed form'' solutions which,
after lots of automatic and manual simplifications, became the
formulae for $T(k)$, $S_i(n)$, and $P_i(n)$.

\subsection*{Acknowledgments}
We would like to thank Christian Krattenthaler and Peter Paule for
their many valuable comments, Manuel Kauers for his support concerning
the guessing software, Doron Zeilberger for his encouragement to
tackle these conjectures, and the two anonymous referees for their
diligent work.

\bibliographystyle{plain}

\end{document}